\documentclass[11pt]{article}
\usepackage[utf8]{inputenc}
\usepackage[T1]{fontenc}
\usepackage[top=1in, bottom=1in, left=1in, right=1in, letterpaper]{geometry}

\usepackage{amsmath, amssymb, amsfonts, bm, cite}
\usepackage{amsthm}

\usepackage{enumitem}
\usepackage{comment}
\usepackage{xspace}
\usepackage{ifthen}
\usepackage{boxedminipage}
\usepackage{algorithm}
\usepackage{algpseudocode}
\usepackage{color}
\usepackage{empheq}
\usepackage[colorlinks=true,citecolor=blue,bookmarksnumbered=true,hyperfootnotes=true]{hyperref}
\usepackage{cleveref}

\algrenewcomment[1]{\hfill{\scriptsize//~#1}}

\newtheorem{theorem}{Theorem}
\newtheorem{lemma}[theorem]{Lemma}

\newtheorem{corollary}[theorem]{Corollary}

\usepackage{color}
\newcommand{\Aviad}[1]{}
\newcommand{\Jan}[1]{}

\theoremstyle{remark}

\theoremstyle{definition}
\newtheorem{definition}[theorem]{Definition}

\crefname{theorem}{Theorem}{Theorems}
\crefname{lemma}{Lemma}{Lemmas}
\crefname{proposition}{Proposition}{Propositions}
\crefname{corollary}{Corollary}{Corollaries}
\crefname{fact}{Fact}{Facts}
\crefname{definition}{Definition}{Definitions}
\crefname{remark}{Remark}{Remarks}

\crefname{section}{Section}{Sections}
\crefname{appendix}{Appendix}{Appendices}
\crefname{algorithm}{Algorithm}{Algorithms}

\newcommand{\RR}{{\mathbb R}}

\newcommand{\cF}{\mathcal{F}}
\newcommand{\cH}{\mathcal{H}}

\newcommand{\cA}{\mathcal{A}}

\newcommand{\cG}{\mathcal{I}}

\newcommand{\cS}{\mathcal{S}}

\newcommand{\eps}{\epsilon}
\newcommand{\E}{{\mathbb E}}

\newcommand{\vect}[1]{\mathbf{#1}}
\newcommand{\bx}{\vect{x}}

\newcommand{\bp}{\vect{p}}
\newcommand{\by}{\vect{y}}
\newcommand{\bz}{\vect{z}}

\newcommand{\NSW}{\mbox{\sf NSW}}

\newcommand{\med}{\mbox{\sf med}}

\def\b1{{\bf 1}}

\title{A constant factor approximation for Nash social welfare with subadditive valuations}

\author{Shahar Dobzinski\thanks{Weizmann Institute of Science}, Wenzheng Li\thanks{Stanford University}, Aviad Rubinstein\thanks{Stanford University}, Jan Vondr\'{a}k\thanks{Stanford University}}

\begin{document} 
\maketitle \thispagestyle{empty} 

\begin{abstract}
We present a constant-factor approximation algorithm for the Nash social welfare maximization problem with subadditive valuations accessible via  demand queries. More generally, we propose a template for NSW optimization by solving a configuration-type LP and using a rounding procedure for (utilitarian) social welfare as a blackbox, which 
could be applicable to other variants of the problem.
\end{abstract}



\section{Introduction}

We consider the problem of allocating a set $\cG$ of $m$ indivisible items to a set $\cA$ of $n$ agents, where each agent $i\in\cA$ has a valuation function $v_i: 2^{\cG}\to \RR_{\ge0}$. The Nash social welfare (NSW) problem is to find an allocation ${\mathcal{S}}=(S_i)_{i\in \cA}$ that maximizes the geometric mean of the agents' valuations,
\[ \NSW({\mathcal{S}}) = \left( \prod_{i\in \cA} v_i(S_i) \right)^{1/|\cA|}.\]
For $\alpha \ge 1$, an \emph{$\alpha$-approximate solution} to the NSW problem is an allocation $\cS$ with $\NSW(\cS)\ge \mathrm{OPT}/\alpha$, where $\mathrm{OPT}$ denotes the optimum value of the NSW-maximization problem.

Allocating resources to agents in a fair and efficient manner
is a fundamental problem in computer science, economics, and social choice theory, with substantial prior work \cite{Barbanel,Brams1996,BrandtCELP16,Moulin2004,Robertson1998,R16,Young1995}.
A common measure of efficiency is \emph{utilitarian social welfare}, i.e., the sum of the utilities  $\sum_{i\in \cA} v_i(S_i)$ for an allocation $(S_i)_{i\in \cA}$. This objective does not take fairness into account, as all items could be allocated to one agent whose valuation function dominates the others.
In order to incorporate fairness, various notions have been considered, ranging from envy-freeness, proportional fairness to various modifications of the objective function. At the end of the spectrum opposite to utilitarian social welfare, one can consider the max-min objective, $\min_{i\in\cA} v_i(S_i)$, also known as the \emph{Santa Claus problem} \cite{bansal2006santa}. This objective is somewhat extreme in considering only the happiness of the least happy agent.

Nash social welfare provides a balanced tradeoff between the requirements of fairness and efficiency.  It has been introduced independently in several contexts: as a discrete variant of the Nash bargaining game~\cite{Kaneko1979,nash1950bargaining}; as a notion of competitive equilibrium with equal incomes in economics~\cite{varian1973equity}; and also as a proportional fairness notion in networking~\cite{kelly1997charging}. 
Nash social welfare has several desirable features, for example invariance under scaling of the valuation functions $v_i$ by independent factors $\lambda_i$, i.e., each agent can express their preference in a ``different currency'' without changing the optimization problem (see~\cite{Moulin2004} for additional characteristics).

\subsection{Preliminaries}

The difficulty of optimizing Nash social welfare depends naturally on the class of valuation functions that we want to deal with, and how they are accessible. Various classes of valuations have been considered in the literature. For the sake of this paper, let us introduce three basic classes of valuations, and two oracle models.

\paragraph{Classes of valuation functions}
A set function $v:\, 2^{\cG} \to \RR$  is \emph{monotone} if $v(S)\le v(T)$ whenever $S\subseteq T$. A monotone set function with $v(\emptyset) = 0$ is also called a \emph{valuation function} or simply {\em valuation}.

A valuation $v:\, 2^{\cG} \to \RR$ is \emph{additive} if $v(S) = \sum_{j \in S} w_j$ for nonnegative weights $w_j$.

A valuation $v:\, 2^{\cG} \to \RR$ is \emph{submodular} if
$v(S) + v(T) \ge v(S\cap T) + v(S\cup T)\, \quad  \forall S, T \subseteq   \cG$.

A valuation $v:\, 2^{\cG} \to \RR$ is \emph{fractionally subadditive} (or {\em XOS}) if
$v(S) = \max_{i \in \cG} \sum_{j \in S} w_{ij}$ for nonnegative weights $w_{ij}$.

A valuation $v:\, 2^{\cG} \to \RR$ is \emph{subadditive} if  
$v(S) + v(T) \ge v(S\cup T)\, \quad  \forall S, T \subseteq   \cG$.

We remark that these classes form a chain of inclusions: additive valuations are submodular, submodular valuations are XOS, and XOS valuations are subadditive.

\paragraph{Oracle access.}
Note that additive valuations can be presented explicitly on the input. However, for more general classes of valuations, we need to resort to oracle access, since presenting a valuation explicitly would take an exponential amount of space. Three types of oracles to access valuation functions have been commonly considered in the literature.
\begin{itemize}
    \item {\em Value oracle:}  Given a set $S \subseteq \cG$, return the value $v(S)$.
    \item {\em Demand oracle:} Given prices $(p_j: j \in \cG)$, return a set $S$ maximizing $v(S) - \sum_{j \in S} p_j$.
    \item {\em XOS oracle} (for an XOS valuation $v$): Given a set $S$, return an additive function $a$ from the XOS representation of $v$ such that $v(S) = a(S)$.
\end{itemize}

\subsection{Prior work}


The Nash social welfare problem is NP-hard already in the case of two agents with identical additive valuations, by a reduction from the Subset-Sum problem. For multiple agents, it is NP-hard to approximate within a factor better than $0.936$ for additive valuations~\cite{garg2017satiation}, and better than $1-1/e \simeq 0.632$ for submodular valuations~\cite{GargKK20}.

The first constant-factor approximation algorithm for additive valuations, with the factor of $1/(2e^{1/e}) \approx 0.346$, was given by Cole and Gkatzelis~\cite{cole2015approximating} using a continuous relaxation based on a particular market equilibrium concept. Later,~\cite{cole2017convex} improved the analysis of this algorithm to achieve the factor of $1/2$. 
Anari, Oveis Gharan, Saberi, and Singh~\cite{anari2017nash} used a convex relaxation that relies on properties of real stable polynomials, to give an elegant analysis of an algorithm that gives a factor of $1/e$. The current best factor is $1/e^{1/e}-\eps \simeq 0.692$ by Barman, Krishnamurthy, and Vaish~\cite{barman2018finding}; the algorithm uses a different market equilibrium based approach. Note also that this factor is above $1-1/e$, hence separating the additive and submodular settings.

Constant-factor approximations have been extended to some classes beyond additive functions: capped-additive~\cite{garg2018approximating}, separable piecewise-linear concave (SPLC)~\cite{anari2018nash}, and their common generalization, capped-SPLC~\cite{ChaudhuryCGGHM18} valuations; the approximation factor for capped-SPLC valuations  matches the $1/e^{1/e}-\varepsilon$ factor for additive valuations. All these valuations are special classes of submodular ones. 
Subsequently, Li and Vondr\'ak~\cite{LiV20} designed an algorithm that estimates the optimal value within a factor of $\frac{(e-1)^2}{e^3} \simeq 0.147$ for a broad class of submodular valuations, such as coverage and summations of matroid rank functions, by extending the techniques of~\cite{anari2017nash} using real stable polynomials. However, this algorithm only estimates the optimum value but does not find a corresponding allocation in polynomial time.

 An important concenptual advance was presented in \cite{GHV20}, where a relaxation combining ideas from matching theory and convex optimization was shown to give a constant factor for the class of ``Rado valuations'' (containing weighted matroid rank functions and some related valuations). A crucial property of this approach is that it is quite modular and ended up leadin to multiple further advances.  In \cite{LiV21}, this approach was extended to provide a constant factor approximation algorithm for general submodular valuations, by replacing the concave extension of a valuation by the multilinear extension. The initial factor was rather small ($1/380$). Recently, a much simpler algorithm combining matching and local search was presented to give a $(1/4-\eps)$-approximation for submodular valuations \cite{GHLVV23}.

For the more general classes of XOS and subadditive 
valuations~\cite{BBKS20,chaudhury2021fair,GargKK20}, however, only polynomial approximation  factors were known until now, and this is the best one can hope for in the value oracle model~\cite{BBKS20}, for the same reasons that this is a barrier for the utilitarian social welfare problem~\cite{DNS10}. The best known approximation factors up to now have been $O(1/n)$ for subadditive valuations, and $O(1/n^{53/54})$ for XOS valuations if we are given access to both demand and XOS oracles~\cite{barman2021sublinear}. 
Constant factors for XOS valuations seemed quite out of reach prior to this work, and obtaining any sublinear factor for subadditive valuations was stated as an open problem in \cite{barman2021sublinear}.

\Aviad{Any of the previous papers that explicitly mention constant factor for XOS or subadditive as open problems?}
\Jan{Good point. The previous $n^{53/54}$ paper states obtaining any sublinear approximation for subadditive valuations as an open problem. I added a sentence above.}

\subsection{Our results and techniques}

\Aviad{It would be nice to at least include informal theorem statements in the intro}
\Jan{OK.}

Our main result is the following.

\paragraph{Theorem.} (informal)
{\em There is an algorithm, with polynomial time and demand queries to agents' valuations, that provides a constant factor approximation for the Nash Social Welfare with subadditive valuations.}

\Aviad{Is it true that we can run in polynomial time given demand queries?}
\Jan{Yes, I believe that's the result.}

As a special case, this also gives a constant-factor approximation for XOS valuations accessible via demand queries. (The algorithm for XOS valuations is somewhat simpler, as we discuss later in this section.)
This completes the picture in the sense that now we have a constant factor approximation for Nash social welfare in the main settings where one is known for utilitarian social welfare: for submodular valuations with value queries, and for subadditive valuations with demand queries. (It is known that a stronger oracle than a value oracle is required for XOS and subadditive valuations, even for utilitarian social welfare.)

The basis of our approach is the matching+relaxation paradigm which gave a constant-factor approximation for submodular valuations \cite{GHV20,LiV21}. Considering that the only constant-factor approximation for social welfare with subadditive valuations \cite{Feige08} is based on the "Configuration LP", which can be solved using demand queries, it is a natural idea to use a relaxation similar to the Configuration LP. A natural variant for Nash social welfare is the Eisenberg-Gale relaxation, using the logarithm of the concave extension of each agent's valuation.
We apply this relaxation on top of an initial matching, as in \cite{GHV20}.

The main obstacle with this approach is that natural rounding procedures for the Configuration LP do not satisfy any concentration properties. At a high level, without concentration some agents have higher value but some have lower value - leading to poor Nash social welfare even if we can maintain the expected utilitarian social welfare. More specifically, the first challenge is that, given a fractional solution $x_{i,S}$, we would ideally like to round it to an integral allocation by allocating set $S$ to agent $i$ with probability $x_{i,S}$. Even though this ideal rounding preserves each agent's expected value, the variance can be arbitrary, depending on the fractional solution $x_{i,S}$. 
Our first technical contribution is a procedure (see Lemma~\ref{lemma:splitting}) for finding a new feasible solution to the Configuration LP that, for each agent, has only high value subsets in its support (with the exception of agents who get most of their value from a single item --- this case is handled separately with the matching procedure). 
This procedure is rather simple in hindsight. At a high level, we can think of the fractional solution as a distribution of allocations for each agent. We want to discard the part of the distribution that corresponds to low value subsets; but this leaves uncovered probability mass. We re-cover this remaining mass by splitting high value subsets to ``stretch'' over more probability mass, while allocating each item with to agent $i$ with the same total probability. 



The next obstacle in rounding the Configuration LP is resolving ``contentions'': aka under the ideal rounding procedure described above, we may be trying to allocate the same item to multiple agents (even though in expectation it is only allocated to one agent). For XOS valuations, a simple independent randomized contention resolution scheme guarantees a constant factor approximation and also enjoys good concentration.  However the situation is more complicated for subadditive valuations. The only known constant-factor approximation for social welfare with subadditive valuations is a rather intricate rounding procedure of Feige \cite{Feige08}, which does not seem to satisfy any useful concentration properties. In any rounded solution, there might be agents who receive very low value, which hurts Nash social welfare, and hence we cannot use it directly.

Our solution is an iterated rounding procedure, where in each stage a certain fraction of agents is ``satisfied'' in the sense that they receive value comparable to their fractional value. We allocate the respective items to them, subject to random filtering which ensures that enough items are still left for the remaining agents. Then we recurse on the remaining agents and remaining items.
Still, some agents may receive a relatively small value, but we guarantee that the fraction of agents who receive low values is proportionally small, which means that the Nash social welfare overall is guaranteed to be good. As an example: if $OPT = (V_1 \cdots V_n)^{1/n}$, it suffices to solve for  an allocation where $\frac{n}{2}$ agents receive value at least $\frac12 V_i$, $\frac{n}{4}$ agents receive value at least $\frac14 V_i$, $\frac{n}{8}$ agents receive value at least $\frac18 V_i$, and so on. Then the approximation factor in terms of Nash Social Welfare turns out to be
$$ (1/2)^{1/2} (1/4)^{1/4} (1/8)^{1/8} (1/16)^{1/16} \cdots $$
and this infinite product converges to $1/4$ (we leave this as an exercise for the reader).

In order to guarantee the success of this rounding procedure, we need a concentration inequality (as in previous works). Concentration properties of subadditive functions are somewhat weaker and more difficult to prove that for submodular or XOS functions. Here we appeal to a powerful subadditive concentration inequality presented by Schechtman \cite{Schechtman99}, which is based on the ``$q$-point control inequality'' of Talagrand \cite{Talagrand89,Talagrand95}.

\subsubsection*{Reducing Nash Welfare to Rounding of Configuration LP}

Technically, we prove a reduction theorem (Theorem~\ref{thm:reduction}) which shows that to achieve a constant factor approximation for Nash social welfare, is it sufficient to implement efficient subroutines for two subproblems: (1) finding a solution of the Configuration LP satisfying a certain additional property (which happens to be satisfied for example by an optimal solution of the Eisenberg-Gale relaxation \cite{GHV20}, or the continuous greedy algorithm for the log-multilinear relaxation \cite{LiV21}), and (2) rounding a fractional solution of the Configuration LP while losing only a constant factor with respect to {\em utilitarian social welfare}. The latter problem is relatively easy for XOS valuations, but non-trivial for subadditive valuations. Fortunately, a factor $1/2$ rounding procedure is known due to Feige's work on welfare maximization with subadditive bidders \cite{Feige08}, which we use here as a blackbox. 

We remark that the constant factors lost in various stages of our proof are rather large and lead to a final approximation factor of $\sim 1/375,000$ for the Nash social welfare problem with subadditive valuations. One may hope that as in the case of submodular valuations, an initially large constant factor can be eventually improved to a ``practical one''.

\paragraph{Paper organization.}
In Section~\ref{sec:reduction}, we present the main technical result, which is a reduction of Nash social welfare to a certain relaxation solver and a rounding procedure for the Configuration LP. In Section~\ref{sec:subadditive}, we show how this implies a constant factor approximation algorithm for Nash social welfare with subadditive valuations.

We defer some components of the algorithm which are similar to earlier work to the appendices: Solving the relaxation and proving the required guarantees (Appendix~\ref{app:solving-log-concave}), the rematching lemmas (Appendix~\ref{app:rematching}), and concentration of subadditive functions (Appendix~\ref{app:concentration}).

\section{Optimizing NSW via relaxation and rounding for social welfare}
\label{sec:reduction}

Here we describe our general approach which allows us to derive algorithms for NSW optimization in several settings.
At a high-level, we reduce the NSW optimization to finding a certain solution for the "Configuration LP" (for social welfare optimization), and having a rounding procedure for the Configuration LP, again with respect to social welfare.

Let us define the Configuration LP:

\begin{align*}
\tag{Configuration LP}
\label{eq:config-LP}
\max \quad & \sum_{i\in \cA} \sum_{S \subseteq \cG} v_i(S) x_{i,S} \\
& \quad \sum_{i \in \cA} \sum_{S \subseteq \cG: j \in S} x_{i,S} \le 1 && \forall j \in \cG \\
& \quad \sum_{S \subseteq \cG} x_{i,S} = 1 && \forall i \in \cA \\
& \quad x_{i,S} \ge 0 &&  \forall i \in \cA, S \subseteq \cG \\
\end{align*}

Equivalently, this can be written as 
\begin{align*}
\max \quad & \sum_{i\in \cA} v_i^+(\bx_i) \\
& \quad \sum_{i \in \cA} x_{ij} \le 1 && \forall j \in \cG \\
& \quad x_{ij} \ge 0 &&  \forall i \in \cA, j \in \cG \\
\end{align*}

where as before,
\begin{align*}
\tag{Concave Extension}
 v_i^+(\bx_i) = & \max \sum_{S\subset \cG} v_i(S) x_{i,S} : \\
 &  \sum_{S \subset \cG:j\in S} x_{i,S} \le x_{ij} && \forall j\in \cG \\
 & \ \ \sum_{S\subset\cG} x_{i,S} = 1 \\
 & \ \  x_{i, S} \ge 0 &&  \forall S\subset\cG
\end{align*}

The following is our main reduction theorem, which provides an algorithm for Nash social welfare, given two procedure that we call the {\bf Relaxation Solver} and {\bf Rounding Procedure}. Note that assumption on the {\bf Relaxation Solver} is somewhat unusual: It is not that $(x_{i,S})$ is an optimal or near-optimal solution of (\ref{eq:config-LP}), but a different condition that the optimum social welfare with valuations by $w_i(S) = v_i(S) / V_i = v_i(S) / \sum_{S'} v_i(S') x_{i,S'}$ is upper-bounded by $c|\cA|$. (The social welfare of $x_{i,S}$ itself with valuations $w_i$ is exactly $|\cA|$, so as a consequence $(x_{i,S})$ is $c$-approximate optimum with respect to the valuations $w_i$.)
This condition is required primarily for the later ``rematching'' step (Lemma~\ref{lemma:extension}). 
Fortunately, this condition is satisfied by natural approaches to solve the ``Gale-Eisenberg'' relaxation, which replaces the continuous valuation extensions by their logarithms. We discuss this further in Section~\ref{sec:subadditive}.

\begin{theorem}
\label{thm:reduction}
Suppose that for a certain class of instances of Nash social welfare, with subadditive valuations, we have the following procedures available, with parameters $c,d \geq 1$:
\begin{itemize}
    \item {\bf Relaxation Solver:} Given valuations $(v_i: i \in \cA)$ on a set of items $\cG$, we can find a feasible solution $(x_{i,S})$ of (\ref{eq:config-LP}) such that the social welfare optimum with valuations
    $$ w_i(S) = \frac{1}{V_i} v_i(S), \ \ \ \ \    V_i = \sum_{S \subseteq \cG} v_i(S) x_{i,S} $$
    is at most $c |\cA|$.

    \item {\bf Rounding Procedure:} Given a feasible solution $(x_{i,S})$ of (\ref{eq:config-LP}), we can find an allocation $(S_1,\ldots,S_n)$ where each $S_i$ is a subset of some set $S'_i$ such that $x_{i,S'_i} > 0$ and
    $$ \sum_{i \in \cA} w_i(S_i) = \sum_{i \in \cA} \frac{1}{V_i} v_i(S_i) \geq \frac{1}{d} |\cA|.$$
    (As above, $V_i = \sum_{S \subseteq \cG} v_i(S) x_{i,S}$.)
    \end{itemize}

Then there is an algorithm which provides an $O(cd^2)$-approximation in Nash social welfare for the same class of instances, using $1$ call to the Relaxation Solver and a logarithmic number of calls to the Rounding Procedure. The running time is polynomial in $|\cA|, |\cG|$ and the support of the fractional solution $(x_{i,S})$. 
\end{theorem}

  In the following, we prove this theorem by presenting an algorithm with several phases. These phases are similar to recent matching-based algorithms for Nash social welfare \cite{GHV20,GargHV21,LiV21,GHLVV23} with the exception of two phases which are new (phases 3,4 below). The high-level outline is as follows.

\paragraph{NSW Algorithm Template.}
\begin{enumerate}
    \item  We find an initial matching $\tau:\cA \to \cG$, maximizing $\prod_{i \in \cA} v_i(\{\tau(i)\})$. Let $\cH = \tau[\cA]$ denote the matching items and $\cG' = \cG \setminus \cH$ the remaining items. Let also $\cA' = \{ i \in \cA: v_i(\cG') > 0\}$. 

    \item We apply the {\bf Relaxation Solver} to obtain a fractional solution $(x_{i,S})_{i \in \cA', S \subseteq \cG'}$ and values $V_i = \sum_{S \subseteq \cG'} v_i(S) x_{i,S}$. We can view these values as ``targets'' for different agents to achieve. 

    \item Let $\nu_i = \max_{i \in \cG'} v_i(j)$ and $\cA'' = \{ i \in \cA': V_i \geq 6 \nu_i \}$. We preprocess the fractional solution $(x_{i,S})$ for $i \in \cA''$, removing sets of low value and partitioning sets of high value, so that for every set in the support of the new fractional solution $x'_{i,S}$ for agent $i$, we have $v_i(S) = \Theta(V_i)$.

    \item We apply the {\bf Rounding Procedure} to $x'_{i,S}$ to find an allocation $(S_i \in \cA'')$ satisfying
    $$\sum_{i \in \cA} \frac{1}{V_i} v_i(S_i) = \Omega\left(\frac{1}{d} |\cA|\right).$$
    Since each $S_i$ has value at most $V_i$ (due to our preprocessing), it must be the case that a $\Theta(\frac{1}{d}$)-fraction of agents receive value at least $\Theta(\frac{1}{d} V_i)$. We allocate a random $\Theta(\frac{1}{d})$-fraction of items to this $\Theta(\frac{1}{d})$-fraction of agents (each item from their respective sets independently with probability $\Theta(\frac{1}{d})$); call the resulting set $T_i$ for agent $i$. We repeat this phase for the remaining items and agents, until there are no agents left. For agents $i \in \cA \setminus \cA''$, we define $T_i = \emptyset$.

    \item We recompute the initial matching to obtain a new matching $\sigma:\cA \to \cH$, which maximizes $\prod_{i \in \cA} v_i(T_i + \sigma(i))$. We allocate $T_i + \sigma(i)$ to agent $i$.

\end{enumerate}

Now we proceed to analyze the phases of this algorithm more rigorously.

\subsection{Initial Matching}

There is nothing new in this phase. We can find a matching $\tau:\cA \to \cG$ maximizing $\prod_{i \in \cA} v_i(\tau(i))$ by solving a max-weight matching problem with edges $(i,j)$ where $v_i(j)>0$, and weights $w_{ij} = \log v_i(j)$. 

We denote by $\cH = \tau[\cA]$ the matched items, by $\cG' = \cG \setminus \cH$ the remaining items, and by $\cA' = \{ i \in \cA: v_i(\cG') > 0 \}$ the agents who get positive value from $\cG'$.

A property we need in the following is the following.

\begin{lemma}
If $\tau: \cA \to \cG$ is a matching maximizing $\prod_{i \in \cA} v_i(\tau(i))$ then for any $j \in \cG' = \cG \setminus \tau[\cA]$, $v_i(j) \leq v_i(\tau(i))$.
\end{lemma}

\begin{proof}
If there is $j \in \cG'$, $v_i(j) > v_i(\tau(i))$, then we can swap $\tau(i)$ for $j$ in the matching and increase its value. 
\end{proof}

For subadditive valuations, we also get $v_i(S+j) - v_i(S) \leq v_i(\tau(i))$ for any $S \subset \cG', j \in \cG' \setminus S$ (since $v_i(S+j) \leq v_i(S) + v_i(j)$).

\subsection{Relaxation Solver}
\label{sec:rematching}

Here we assume that the {\bf Relaxation Solver} is available as a black-box. We return to its implementations in specific settings in Section~\ref{sec:subadditive}.

We apply the {\bf Relaxation Solver} to the residual instance on items $\cG' = \cG \setminus \cH$ and agents $\cA'$ who have nonzero value for some items in $\cG'$. 
The important property of the obtained solution $(x_{i,S})$ is that after scaling the valuations as follows,
 $$ w_i(S) = \frac{1}{V_i} v_i(S), \ \ \ \ \    V_i = \sum_{S \subseteq \cG'} v_i(S) x_{i,S} $$
the social welfare optimum for $w_1,\ldots,w_n$ is at most $c|\cA'|$. In other words, for any feasible allocation $(T^*_1,\ldots,T^*_n)$ of $\cG'$, we have
$$ \sum_{i \in \cA'} \frac{v_i(T^*_i)}{V_i} \leq c |\cA'|.$$

\subsection{Set Splitting}

Here we describe Phase 3, preprocessing of the fractional solution. We will work only with agents who get significant value from the fractional solution:
Let $\nu_i = \max_{j \in \cG'} v_i(j)$ and
$$ \cA'' := \{ i \in \cA': V_i \geq 6 \nu_i \}. $$
We prove the following.

\begin{lemma}
\label{lemma:splitting}
Assume that the valuations $v_1,\ldots,v_n$ are subadditive. 
Given a feasible solution $(x_{i,S})$ of (\ref{eq:config-LP}) for an instance with agents $\cA''$ and items $\cG'$, where $V_i = \sum_{S \subseteq \cG'} v_i(S) x_{i,S}$ and $\nu_i = \max_{j \in \cG'} v_i(j)$, we can find (in running time and a number of value queries polynomial  in the number of nonzero coefficients $x_{i,S}$) a modified solution $(x'_{i,S})$ such that
\begin{itemize}
    \item For every $S$ such that $x'_{i,S} > 0$, $\frac13 V_i - \nu_i \leq v_i(S) \leq V_i$. \item For every $i \in \cA''$, $\sum_{S \subseteq \cG} x'_{i,S} = 1$.
    \item For every $j \in \cG'$, $\sum_{i,S \ni j} x'_{i,S} \leq 1$.
    
\end{itemize} 
\end{lemma}

\begin{proof}
We apply the following procedure to the fractional solution $\bx = (x_{i,S})$.

\paragraph{SetSplitting($\bx$).}
\begin{enumerate}
     \item Let  $V_i = \sum_{S \subseteq \cG'} v_i(S) x_{i,S}$, $\nu_i = \max_{j \in \cG'} v_i(j)$, and $\cF_i = \{S\subseteq\cG': v_i(S)\ge \frac13 V_i \}$. 

    \item Set $x'_{i, S} = 0$ and $k_{i,S} = 0$  for $S \notin \cF_i$; i.e., discard sets whose value is too low.
    
    \item For every $S \in \cF_i$, let $k_{i, S} = \lfloor \frac{3 v_{i}(S)}{V_i} \rfloor$. Split $S$ into sets $S_1, \cdots, S_{k_{i, S}}$ such that $\forall \ell=1,\ldots,k_{i,S}$,
    $$ v_i(S_\ell) \ge \frac13 V_i - \nu_i.$$
    Note that this is possible since by subadditivity, the average value of a subset in any partition of $S$ into $k_{i,S}$ subsets is at least $v_i(S) / k_{i,S} \geq \frac13 V_i$, and indivisibility of items can cause the value to drop by at most $\nu_i$. 

    \item For each set $S_\ell$ produced above, remove some items if necessary to ensure that its value is at most $V_i$. Call the resulting set $S'_\ell$. Note that since removing an item can decrease the value by at most $\nu_i$, we start from value $\geq \frac13 V_i - \nu_i$, and we only remove items as long as the value is more than $V_i$, we can conclude that
    $$V_i \geq v_i(S'_\ell) \ge \frac13 V_i - \nu_i.$$
    
    \item Set $\tilde{x}_{i,T} = \sum_{S \in \cF_i, \exists \ell: S'_\ell=T} x_{i,S}$, and 
         $x'_{i,T} = \frac{\tilde{x}_{i,T}}{\sum_S \tilde{x}_{i,S}}$.

    \item Return $\bx'$.
    
\end{enumerate}

    Let us now prove the desired properties of $\bx'$. By construction (step 5), the solution is normalized in the sense that $\sum_T x'_{i,T} = 1$ for every $i \in \cA''$. Also, as we argued above, $ V_i \geq v_i(T) \ge \frac13 V_i - \nu_i$ for every set $T$ participating in the support of $\bx'$. 
    It remains to prove that the coefficients $x'_{i,T}$ add up to at most $1$ on each item.
    
    Let us first consider $\tilde{x}_{i,T}$: Since each contribution to $\tilde{x}_{i,T}$ for $j \in T$ is inherited from some coefficient $x_{i,S}$ where $j \in S$, each coefficient $x_{i,S}$ contributes at most once in this way, and the coefficients $x_{i,S}$ for $S \ni j$ add up to at most $1$, it is clear that $\sum_{i,T \ni j} \tilde{x}_{i,T} \leq 1$. Finally, $x'_{i,T}$ is obtained by normalizing $\tilde{x}_{i,T}$; so we need to be concerned about the summation $\sum_{S} \tilde{x}_{i,S}$, which could be possibly less than $1$.
  
    
    We have:
    $$\sum_{S\in \cF_i} v_i(S)x_{i, S} = V_i - \sum_{S \notin \cF_i} v_i(S) x_{i, S}\ge V_i - \frac{1}{3} V_i = \frac23 V_i.$$ 
    Observe that each coefficient $x_{i,S}$ for $S \in \cF_i$ contributes $k_{i,S}$ coefficients of the same value to the summation $\sum_{S} \tilde{x}_{i,S}$, and the union of the respective sets is $S$. 
    So we have
    $$\sum_{S\subseteq \cG'} \tilde{x}_{i,S} = \sum_{S \in \cF_i} k_{i,S} x_{i,S} = \sum_{S \in \cF_i} \left\lfloor\frac{3 v_{i}(S)}{V_i} \right\rfloor \cdot x_{i,S} \ge \frac{3}{2} \sum_{S\in \cF_i}\frac{v_{i}(S)x_{i, S}}{V_i} \geq 1 $$ 
    considering that $3 v_i(S) / V_i \geq 1$ for $S \in \cF_i$, so the floor operation can decrease the ratio by at most a factor of $2$. Also, we have $\sum_{S \in \cF_i} v_i(S) x_{i,S} \geq \frac23 V_i$ from above.
    Hence $x'_{i,T} = \frac{\tilde{x}_{i,T}}{\sum_S \tilde{x}_{i,S}} \leq \tilde{x}_{i,T}$ and the coefficients $x'_{i,S}$ for $S \ni j$ add up to at most $1$.
\end{proof}

\subsection{Iterated Rounding}

Finally, we need to round the fractional solution $(x'_{i,S})$ obtained in the previous phase. 
As a subroutine, we use the assumed {\bf Rounding Procedure} for (additive) social welfare.

Given a fractional solution $\bx' = (x'_{i,S})$ obtained in the previous phase, we call the procedure {NSW-ROUND}$(\bx',\cA'',\cG',\delta)$ with a parameter $\delta = \frac{1}{7d}$, where $d$ is a approximation factor guaranteed by the {\bf Rounding Procedure}.

\begin{algorithm}[htbp]
\caption{\bf Iterated Rounding}
\label{alg:NSW-rounding}
\begin{algorithmic}[1]
  \Procedure{NSW-Round}{$\bx',\cA_0,\cG_0,\delta$}:
    \State Let $V'_i \leftarrow \sum_{S \subseteq \cG'} v_i(S) x'_{i,S}$
    \State For each item $j \in \cG_0$ independently, let $r_j \leftarrow t$ with probability $\delta (1-\delta)^{t-1}$ for $t \geq 1$.
    \State Let $R_t \leftarrow \{ j \in \cG_0: r_j=t \}$ for all $t \geq 1$
    \State Let $t \leftarrow 1$
    \While{$\cA_t \neq \emptyset$}
    \State $(S_i: i \in \cA_t) \leftarrow {\bf RoundingProcedure}(\bx',\cA_t)$
    \State $\cA_{t+1} \leftarrow \{ i \in \cA_t: v_i(S_i) < \delta V'_i\}$
    \State For each agent $i \in \cA_t \setminus \cA_{t+1}$, allocate $T_i \leftarrow S_i \cap R_t$.
    \EndWhile
    \State Return $(T_i: i \in \cA_0)$
  \EndProcedure
\end{algorithmic}
\end{algorithm}

As we mentioned above, the intuition behind this rounding procedure is that it gives good value to a large fraction of agents, and exponentially small values to an exponentially decaying number of agents, so overall its Nash social welfare is good. We prove this in a sequence of lemmas.

\begin{lemma}
Under our assumption on the {\bf Rounding Procedure}, and setting $\delta = \frac{1}{7d}$, in each round there is at least a $\delta$-fraction of agents (rounded up to the nearest integer) who receive value at least $\delta V'_i$.
\end{lemma}

\begin{proof}
Note that $V'_i \geq \frac16 V_i$, since every set in the support of $x'_{i,S}$ has value at least $\frac13 V_i - \nu_i \geq \frac16 V_i$. We assume that under valuations $w_i(S) = \frac{1}{V'_i} v_i(S)$, the {\bf Rounding Procedure} returns an allocation $(S_i: i \in \cA_t)$ such that $\sum_{i \in \cA_t} w_i(S_i) \geq \frac{1}{d} |\cA_t|$. Also, the fractional solution $\bx'$ has been processed so that no set in its support for agent $i$ has value more than $V_i \leq 6 V'_i$, and the rounding only allocates subsets of sets in the support of $\bx'$. Hence, we have $w_i(S_i) = \frac{1}{V'_i} v_i(S_i) \leq 6$ for every agent $i$. Consider the agents who receive value $w_i(S_i) \geq \delta$; if the number of such agents is less than $\delta |\cA_t|$, then the total value collected by the agents is $\sum_{i \in \cA_t} w_i(S) < 6 \cdot \delta |\cA_t| + \delta \cdot (1 - \delta) |\cA_t| < 7\delta |\cA_t| = \frac{1}{d} |\cA_t|$, which is a contradiction.
\end{proof}

\begin{lemma}
\label{lemma:rnd-partition}
If $|\cA_0| = a$ and the agents are ordered by the round in which they received items (and arbitrarily within each round), then the $i$-th agent receives each element of her set $S_i$ independently with probability at least $\delta (1 - \frac{i-1}{a})$.
\end{lemma}

\begin{proof}
Consider the $i$-th agent, and suppose that $i \in \cA_t \setminus \cA_{t+1}$, i.e.~the agent gets items in round $t$. We claim that $a (1-\delta)^{t-1} \geq n-i+1$: In each round, we allocate items to at least a $\delta$-fraction of agents, so the set of agents $\cA_{t-1}$ remaining after $t-1$ rounds has size at most $a (1-\delta)^{t-1}$. This set must include agent $i$, otherwise she would have been satisfied earlier. Therefore, $a-i+1 \leq a (1-\delta)^{t-1}$.

The items allocated to agent $i$ in round $t$ are $S_i \cap R_t$, where $R_t$ contains each element independently with probability $\delta (1-\delta)^{t-1}$. By the argument above, 
$\delta (1-\delta)^{t-1} \geq \delta \cdot \frac{a-i+1}{a}$.
\end{proof}

\begin{lemma}
\label{lemma:ratio-bound}
If $T_i$ is the set allocated to the $i$-th agent in the ordering defined above (and we assume w.l.o.g. that the index of this agent is also $i$), and $\max_{j \in \cG'} v_i(j) \leq \nu_i$ then
$$ \E\left[ \log \frac{V_i}{v_i(T_i) + \nu_i} \right] \leq  \log \frac{60}{\delta^2 (1-\frac{i-1}{n})}.$$
\end{lemma}

\begin{proof}
By definition, the set $S_i$ tentatively chosen for the $i$-th agent in the round where $i \in \cA_t \setminus \cA_{t+1}$ satisfies $$v_i(S_i) \geq \delta V'_i \geq \delta \left(\frac{1}{3} V_i - \nu_i \right) 
 $$ (see Lemma~\ref{lemma:splitting}). By Lemma~\ref{lemma:rnd-partition}, the $i$-th agent receives a set $T_i = S_i \cap R_t$ which contains each element of $S_i$ independently with probability at least $\delta (1-\frac{i-1}{n})$.

Consider now the expression $\log \frac{V_i}{f(T_i)}$, where $f(T_i) = v_i(T_i) + \nu_i$. This is a random quantity due to the randomness in $R_t$ (the set $S_i$ is fixed here). We use concentration of subadditive functions (Theorem~\ref{thm:lower-tail}) to argue that this expression is not too large in expectation. 
We have $f(S_i) = v_i(S_i) + \nu_i \geq \frac13 \delta V_i$. By the expectation property of subadditive functions (Lemma~\ref{lemma:subadditive-exp}), we have 
\begin{gather*}\E[f(T_i)] = \E[f(S_i \cap R_t)] \geq \delta (1 - \frac{i-1}{n}) \cdot \frac13 \delta V_i = \frac13 \delta^2 (1 - \frac{i-1}{n}) V_i.
\end{gather*}
Let us denote the last expression $\mu_i := \frac13 \delta^2 (1 - \frac{i-1}{n}) V_i \leq \E[f(T_i)]$. 

Now, we apply the lower-tail inequality, Theorem~\ref{thm:lower-tail}, with $q=2$.
Observe that we can assume $\nu_i < \frac{1}{20} \mu_i$. Otherwise,
$$ \frac{V_i}{\nu_i} \leq \frac{20 V_i}{\mu_i} = \frac{60}{\delta^2 (1 - \frac{i-1}{n})}$$
and so the desired bound holds.

Let us set $q=2$ and $k+1 = \lfloor \frac{\mu_i}{10 \nu_i} \rfloor \geq \frac{\mu_i}{15 \nu_i}$ (considering that $\frac{\mu_i}{10 \nu_i} \geq 2$) in Theorem~\ref{thm:lower-tail}. We get
\begin{align*} \Pr\left[ f(T_i) < \frac{\mu_i}{30} \right]  & \leq
\Pr\left[f(T_i) \leq \frac{\E[f(T_i)]}{15} - \frac{\mu_i}{30} \right] \\
 & \leq \Pr\left[f(T_i) \leq \frac{\E[f(T_i)]}{5(q+1)} - \frac{(k+1) \nu_i}{q+1} \right] \\
& \leq \left( \frac{2}{2^k} \right)^{1/2} \leq \frac{2}{2^{\mu_i/(30 \nu_i)}}.
\end{align*}
Our goal is to bound the expectation $\E[\log \frac{V_i}{f(T_i)}]$.
We distinguish two cases: When $f(T_i) < \frac{1}{30} \mu_i$, we use the bound $f(T_i) \geq \nu_i$, which always holds. Otherwise, we use the bound $f(T_i) \geq \frac{1}{30} \mu_i$. From here,
$$ \E\left[\log \frac{V_i}{f(T_i)}\right] \leq 
\Pr\left[f(T_i) < \frac{\mu_i}{30}\right] \cdot \log \frac{V_i}{\nu_i} + \left( 1 - \Pr\left[f(T_i) < \frac{\mu_i}{30}\right] \right) \cdot \log \frac{30 V_i}{\mu_i} $$
$$ = \Pr\left[f(T_i) < \frac{\mu_i}{30}\right] \cdot \log \frac{\mu_i}{30 \nu_i} + \log \frac{30 V_i}{\mu_i}  $$
$$\leq \frac{2}{2^{\mu_i/(30 \nu_i)}} \log \frac{\mu_i}{30 \nu_i} + \log \frac{30 V_i}{\mu_i}.$$
One can verify that the function $\frac{2}{2^{x}} \log x$ is upper-bounded by $\log 2$ for all $x>0$.\footnote{We are using natural logarithms everywhere. \Aviad{Should we switch to $\ln$ everywhere?} \Jan{not sure right now... $\log$ looks generally nicer}}
Hence, 
$$ \E\left[\log \frac{V_i}{f(T_i)}\right] \leq \log 2 + \log \frac{30 V_i}{\mu_i}
< \log \frac{60 V_i}{\mu_i} = \log \frac{60}{\delta^2 (1 - \frac{i-1}{n})}. $$
\end{proof}

\begin{lemma}
\label{lemma:rounding-result}
If $T_i$ is the set allocated to the $i$-th agent in the ordering defined above (and we assume w.l.o.g. that the index of this agent is also $i$), and $\max_{j \in \cG'} v_i(j) \leq \nu_i$ then
$$ \frac{1}{|\cA''|} \sum_{i \in \cA''} \E\left[ \log \frac{V_i}{v_i(T_i) + \nu_i} \right] 
 \leq \log \frac{165}{\delta^2}.
$$
\end{lemma}

\begin{proof}
Let us denote $a = |\cA''|$.
From Lemma~\ref{lemma:ratio-bound}, we have
\begin{eqnarray*}
 \frac{1}{a} \sum_{i=1}^{a} \E\left[ \log  \frac{V_i}{v_i(T_i) + \nu_i} \right] 
& \leq & \frac{1}{a} \sum_{i=1}^{a} \log \frac{60}{\delta^2 (1 - \frac{i-1}{n})} 
 = \left( \log \frac{60}{\delta^2} - \frac{1}{a} \sum_{i=1}^{a}  \log \left(1 - \frac{i-1}{a} \right) \right) 
\end{eqnarray*}
Here, we have $\sum_{i=1}^{a} \log (1 - \frac{i-1}{a}) = \log \prod_{i=1}^{a} \frac{a-i+1}{a} = \log \frac{a!}{a^{a}} \geq -a$ by a standard estimate for the factorial.  So we can conclude
$$  \frac{1}{a} \sum_{i=1}^{a}  \E\left[ \log  \frac{V_i}{v_i(T_i) + \nu_i} \right] \leq  \log \frac{60}{\delta^2} + 1 < \log \frac{165}{\delta^2}. $$
\end{proof}

This concludes the analysis of the iterated rounding phase, which allocates the set $T_i$ to each agent $i \in \cA''$. For agents $i \in \cA \setminus \cA''$, we set $T_i=\emptyset$.

\subsection{Rematching and finishing the analysis}

The last step in the algorithm is to replace the initial matching $\tau: \cA \to \cH$ with a new matching $\sigma: \cA \to \cH$ which is optimal on top of the allocation $(T_i: i \in \cA)$. To analyze this step, we need two lemmas from previous work, whose proofs can be modified easily to yield the following. (We provide full proofs in the appendix.)

\begin{lemma}[matching extension]
\label{lemma:extension}
Let $\tau:\cA \rightarrow \cG$ be the matching maximizing $\prod_{a \in \cA} v_i(\tau(a))$, $\cH = \tau(\cA)$, and $\cG' = \cG \setminus \cH$. Let $(V_i: i \in \cA)$ be values such that $V_i>0$ for $i \in \cA'$, $V_i=0$ for $i \in \cA \setminus \cA'$ and
$$ \sum_{i \in \cA'} \frac{v_i(T^*_i)}{V_i} \leq c |\cA'| $$
for every allocation $(T^*_1,\ldots,T^*_n)$ of the items in $\cG'$.
Then there is a matching $\pi:\cA \rightarrow \cH$ such that 
$$ \prod_{i \in \cA} (V_i + v_i(\pi(i)))^{1/|\cA|}
 \geq \frac{1}{c+1} \prod_{i \in \cA} \left( v_i(S^*_i) \right)^{1/|\cA|} = \frac{OPT}{c+1}$$
 where $(S^*_1,\ldots,S^*_n)$ is an allocation of $\cG$ optimizing Nash social welfare.
\end{lemma}

\begin{lemma}[rematching]
\label{lemma:rematching}
Let $\tau:\cA \rightarrow \cG$ be the matching maximizing $\prod_{a \in \cA} v_i(\tau(a))$, $\cH = \tau(\cA)$, $\cG' = \cG \setminus \cH$, $\pi:\cA \to \cH$ another arbitrary matching, and $\nu_i = \max_{j \in \cG'} v_i(j)$. Let $(W_i: i \in \cA)$ be nonnegative values.  Then there is a matching $\rho: \cA \rightarrow \cH$ such that
$$ \prod_{i \in \cA} (\max\{W_i,v_i(\rho(i))\})^{1/|\cA|} \geq \prod_{i \in \cA} \left(\max \{W_i, v_i(\pi(i)), \nu_i\} \right)^{1/|\cA|}. $$
\end{lemma}

We apply Lemma~\ref{lemma:extension} with the values $V_i = \sum_{S \subseteq \cG'} v_i(S) x_{i,S}$, where $(x_{i,S})$ is the fractional solution returned by {\bf Relaxation Solver}. Due to our assumptions, the condition of Lemma~\ref{lemma:extension} is satisfied and hence there is a matching $\pi:\cA \to \cH$ as described in Lemma~\ref{lemma:extension}: 
$$ \prod_{i \in \cA} (V_i + v_i(\pi(i) \})^{1/|\cA|}  \geq \frac{OPT}{c+1}. $$
From Lemma~\ref{lemma:rounding-result}, we can find with constant probability an assignment $(T_i: i \in \cA'')$ such that
$$ \left( \prod_{i \in \cA''} \frac{V_i}{v(T_i) + \nu_i} \right)^{1/|\cA''|} < \frac{200}{\delta^2} < 10000 \, d^2. $$
(Recall that $\delta = \frac{1}{7d}$, where $d>1$ is the parameter guaranteed by the {\bf Rounding Procedure}.) 

Moreover, we know that $v(T_i) \leq V_i$ and $V_i \geq 6\nu_i$, hence $v(T_i) + \nu_i \leq 2V_i$ for $i \in \cA''$. For agents in $\cA \setminus \cA''$, we have $T_i = \emptyset$ and $V_i \leq 6 \nu_i$. From here, we have
\begin{eqnarray*}
 \prod_{i \in \cA} \frac{V_i + v_i(\pi(i))}{v(T_i) + \nu_i + v_i(\pi(i))} 
 & \leq & \prod_{i \in \cA''} \frac{2V_i + v_i(\pi(i))}{v(T_i) + \nu_i + v_i(\pi(i))} \prod_{i \in \cA \setminus \cA''}  \frac{6 \nu_i + v_i(\pi(i))}{\nu_i + v_i(\pi(i))} \\
& \leq & \prod_{i \in \cA''} \frac{2V_i}{v(T_i) + \nu_i} \prod_{i \in \cA \setminus \cA''} 6
\leq (20000 \, d^2)^{|\cA|}.
\end{eqnarray*}
Finally, we use the rematching Lemma~\ref{lemma:rematching}, with values $v_i(T_i)$:
there exists a matching $\rho: \cA \to \cH$ such that
$$ \prod_{i \in \cA} (\max \{ v_i(T_i), v_i(\rho(i))\})^{1/|\cA|} \geq 
\prod_{i \in \cA} (\max \{ v_i(T_i), v_i(\pi(i)), \nu_i \})^{1/|\cA|} \geq 
\frac13 \prod_{i \in \cA} (V_i + \nu_i + v_i(\pi(i)))^{1/|\cA|} $$
$$ \geq \frac{1}{20000 \, d^2} \prod_{i \in \cA} (V_i + v_i(\pi(i)))^{1/|\cA|}
\geq \frac{OPT}{20000 (c+1) d^2}.$$
Recall that at the end, we find a matching $\sigma:\cA \to \cH$ maximizing $\prod_{i \in \cA} v_i(T_i + \sigma(i))$. Therefore, the NSW value of our solution is at least as much as the one provided by the matching $\rho$, which is $\prod_{i \in \cA} (v_i(T_i + \rho(i))^{1/|\cA|} \geq \prod_{i \in \cA} (\max \{ v_i(T_i), v_i(\rho(i)) \})^{1/|\cA|} \geq \frac{1}{20000 (c+1) d^2} OPT$. This proves Theorem~\ref{thm:reduction}.

\section{Nash social welfare with subadditive valuations}
\label{sec:subadditive}

Here we explain how to use the general framework described in Section~\ref{sec:reduction} to obtain a constant-factor approximation for subadditive valuations, accessible by demand queries.

\begin{theorem}
\label{thm:subadditive}
There is a constant-factor approximation algorithm for the Nash social welfare problem with subadditive valuations, using polynomial running time and a polynomial number of queries to a demand oracle for each agent's valuation.
\end{theorem}

Aside from our general reduction and the ability to solve the Eisenberg-Gale relaxation with demand queries, the main component that we need here is an implementation of a {\bf Rounding Procedure} for subadditive valuations, as described in Theorem~\ref{thm:reduction}. To our knowledge, there is only one such procedure known, which is rather intricate and forms the basis of Feige's ingenious $1/2$-approximation algorithm for maximizing social welfare with subadditive valuations \cite{Feige09}. We use it here as a black-box, which can be described as follows.

\begin{theorem}
\label{thm:Feige-rounding}
For any $\eps>0$, there is a polynomial-time algorithm, which given a fractional solution $(x_{i,S})$ of (\ref{eq:config-LP}) for an instance with subadditive valuations, produces a random allocation $(R_i: i \in \cA)$ such that for every agent, $R_i \subseteq S_i$ for some $S_i, x_{i,S_i} > 0$, and 
$$ \E[v_i(R_i)] \geq \left(\frac12-\eps \right) V_i,   \mbox{       where  } V_i = \sum_{S \subseteq \cG} v_i(S) x_{i,S}.$$
\end{theorem}

For the proof, we refer the reader to Section 3.2.2 of \cite{Feige09}, Theorem 3.9 and the summary of its proof which shows that every player receives expected value at least $(\frac12 - \eps) V_i$.

Now we are ready to prove Theorem~\ref{thm:subadditive}.

\begin{proof}
Considering Theorem~\ref{thm:reduction}, we want to show how to implement the {\bf Relaxation Solver} and {\bf Rounding Procedure} for subadditive valuations.

The {\bf Relaxation Solver} can be obtained applying standard convex optimization techniques to the (\ref{eqns:log-concave}) relaxation. As we discuss in more detail in Appendix~\ref{app:solving-log-concave}, we can compute the values and supergradients of the objective function using demand queries, and obtain an optimal solution satisfying the assumption of Lemma~\ref{cor:concave-opt} (with $f_i = v^+_i$, $\alpha=1)$, and hence 
$$ \sum_{i \in \cA} \frac{v^+_i(\bx^*_i)}{v^+_i(\bx_i)} \leq 2|\cA| $$
for every feasible solution $\bx^*$. Another way to interpret this condition is that for $V_i = v^+_i(\bx_i)$ and modified valuations defined as $w_i(S) = \frac{1}{V_i} v_i(S)$, there is no feasible solution $\bx^*$ achieving value $\sum_{i \in \cA} w^+_i(\bx^*_i) > 2 |\cA|$. In particular, the social welfare optimum with the valuations $(w_i: i \in \cA)$ is at most $2|\cA|$. Hence, we satisfy the {\bf Relaxation Solver}
assumptions with $c = 2$.

Next, we implement the {\bf Rounding Procedure}: Given a fractional solution $(x_{i,S})$, Theorem~\ref{thm:Feige-rounding} gives a procedure which returns a random allocation $(R_i: i \in \cA)$ such that $\E[v_i(R_i)] \geq (\frac12 - \eps) V_i = (\frac12 - \eps) \sum_S x_{i,S} v_i(S)$. This means that for the modified valuations, $w_i(S) = \frac{1}{V_i} v_i(S)$, we have $\E[w_i(R_i)] \geq \frac12 - \eps$, and $\sum_{i \in \cA} w_i(R_i) \geq (\frac12 - \eps) |\cA|$. Hence, we satisfy the {\bf Rounding Procedure} assumptions with $d = \frac{2}{1-2\eps}$. 

Finally, we apply Theorem~\ref{thm:reduction} with $c = 2$ and $d=\frac{2}{1-2\eps}$. We obtain a constant-factor approximation algorithm for the Nash social welfare problem with subadditive valuations accessible via demand queries. (The constant factor ends up being  $20000 (c+1) d^2 = 375000$ for $\eps=0.1$.)
\end{proof}

\bibliographystyle{alpha}
\bibliography{bibfile} 

\appendix

\section{The Eisenberg-Gale Relaxation}
\label{app:solving-log-concave}

We consider the following relaxation of the Nash Social Welfare problem similar to the relaxations in \cite{GargHV21, LiV21}.
\Aviad{Trying to rewrite the following sentence - is this what you meant?} \Jan{Yes, sounds good!} We remark that the application of~\eqref{eqns:log-concave} in the Nash Social Welfare algorithm excludes the items allocated in the initial matching; indeed we ignore those items for the analysis in this section. 

\Aviad{I think it would be better to directly present the general form of Eisenberg-Gale, and then say that for now we plug in $f+i = v_i^+$} \Jan{Okay.}
\begin{align*}
\tag{Eisenberg-Gale Relaxation}
\label{eqns:log-concave}
\max \quad & \sum_{i\in \cA} \log f_i(\bx_i) \\
& \quad \sum_{i \in \cA} x_{ij} \le 1 && \forall j \in \cG \\
& \quad x_{ij} \ge 0 &&  \forall i \in \cA, j \in \cG \\
\end{align*}
where $f_i$ is a suitable relaxation of the valuation function $v_i$ for each $i$.
In particular, we will use the concave extension, $f_i = v_i^+$:
\begin{align*}
\tag{Concave Extension}
\label{eqns:concave}
 v_i^+(\bx_i) := & \max \sum_{S\subseteq\cG} v_i(S) x_{i,S} : \\
 &  \sum_{S \subseteq \cG:j\in S} x_{i,S} \le x_{ij} && \forall j\in \cG \\
 & \ \ \sum_{S\subseteq\cG} x_{i,S} = 1 \\
 & \ \  x_{i, S} \ge 0 &&  \forall S\subseteq\cG
\end{align*}
Note that~\eqref{eqns:concave} is a linear program. The dual LP to (\ref{eqns:concave}) is
\begin{align*}
\tag{Dual LP}
\label{eqns:dual}
 v_i^+(\bx_i) = & \min q + \sum_{j \in \cG} p_j x_{ij}: \\
 &  q + \sum_{j \in S} p_j \ge v_i(S) && \forall S \subseteq \cG \\
  & p_j \ge 0 &&  \forall j \in \cG.
\end{align*}
From here, we can see that $v^+_i(\bx_i)$ is a minimum over a collection of linear functions, and hence a concave function.

\subsection{Solving the Eisenber-Gale Relaxation}
\label{sec:solving-EG}

Here we show how to solve the (\ref{eqns:log-concave}) using demand queries.

\begin{lemma}
\label{lem:solving-log-concave} 
Given demand oracles for $v_1, \cdots, v_n$, an optimal solution $\bx^*$ for (\ref{eqns:log-concave}) can be found within a polynomially small error in polynomial time. Moreover, the support of $\bx^*$ has size polynomial in $n$.
\end{lemma}

Since $v^+_i(\bx_i)$ is a nonnegative concave function, $\log v^+_i(\bx_i)$ is a concave function as well (wherever $v^+_i(\bx_i) > 0$). If we implement the evaluation and supergradient oracles for $\log v^+_i(\bx_i)$, then we can use standard techniques (see, e.g., \cite{DHHW23}) to maximize $\sum_{i \in \cA} \log v^+_i(\bx_i)$ over the convex polytope 
$$ P = \{ \bx \in \RR_+^{\cG \times \cA}: \forall j \in \cG, \sum_{i \in \cA} x_{ij} \leq 1\}.$$

The function $v^+_i(\bx_i)$ can be evaluated with polynomially many demand queries; this is well-known \cite{Feige08} and holds because the demand oracle happens to be the separation oracle for \eqref{eqns:dual}. Hence we can also evaluate $\log v^+_i(\bx_i)$. We focus here on the implementation of a supergradient oracle.

A supergradient of $\log v^+_i$ at a point $\bz$ is any linear function $L_i(\by)$ such that $L_i(\bz) = \log v^+_i(\bz)$ and $L_i(\by) \geq \log v^+_i(\by)$ everywhere. 
Given $\bz$, as a first step, we find a supergradient of $v^+_i$ itself: This can be done by solving the dual LP and finding $\alpha$ and $(p_j: j \in \cG)$ such that $v^+_i(\bz) = \alpha + \sum_{j \in \cG} p_j z_{j} = \alpha + \bp \cdot \bz$.  Since $v^+_i(\by)$ for every $\by$ is the minimum over such linear functions, we also have $v^+_i(\by) \leq  \alpha + \bp \cdot \by$ for all $\by$. Hence $\alpha + \bp \cdot \by$ is the desired supergradient at $\bz$. 

Next, we compute the gradient of $\log (\alpha + \bp \cdot \by)$ with respect to $\by$:
$$ \nabla \log (\alpha + \bp \cdot \by) = \frac{\nabla (\alpha + \bp \cdot \by)}{\alpha + \bp \cdot \by} = \frac{\bp}{\alpha + \bp \cdot \by}.$$
We claim that the linear approximation of $\log (\alpha + \bp \cdot \by)$ obtained by evaluating this gradient at $\bz$,
$$L_i(\by) = \log (\alpha + \bp \cdot \bz) + (\by-\bz) \cdot \nabla (\log (\alpha + \bp \cdot \by))|_{\bz} = (\alpha + \bp \cdot \bz) +  (\by-\bz) \cdot \frac{\bp}{\alpha + \bp \cdot \bz} $$
is a valid supergradient for $\log v^+_i(\by)$ at $\bz$. Indeed, we have $\log v^+_i(\bz)= \log (\alpha + \bp \cdot \bz) = L_i(\bz) $,  and for all $\by$,
$$ \log v^+_i(\by) \leq \log (\alpha + \bp \cdot \by) \leq (\alpha + \bp \cdot \bz) + (\by-\bz) \cdot \nabla (\log (\alpha + \bp \cdot \by))|_{\bz} = L_i(\by). $$
where the second inequality follows from the concavity of $\log (\alpha + \bp \cdot \by)$.

Hence, (\ref{eqns:log-concave}) can be solved in polynomial time, within a polynomially small error, using standard convex optimization techniques \cite{DHHW23}. In particular, we can find a point $\bx$ such that $\sum_{i \in \cA} \log v^+_i(\bx^*_i) \leq \sum_{i \in \cA} \log v^+_i(\bx_i) + \eps$ for every feasible solution $\bx^*$.

Finally, let's explain why the solution can be assumed to have polynomially bounded support. Given a fractional solution $x_{ij}$ (which has obviously polynomially bounded support), for each agent $i$, using demand queries we also obtain a solution of (\ref{eqns:dual}) certifying the value of $v^+_i(\bx_i)$. By complementary slackness, there is a matching primal solution of (\ref{eqns:log-concave}) which has nonzero variables corresponding to the tight constraints in (\ref{eqns:dual}) that define the dual solution. Since the dimension of (\ref{eqns:dual}) is polynomial, the number of such tight constraints is also polynomial. Hence we can assume that the number of nonzero variables in (\ref{eqns:log-concave}) is polynomial.

\subsection{Properties of the optimal solution}

Consider now the (\ref{eqns:log-concave}) in a general form, with objective functions $f_i$ (which could be equal to $v^+_i$ or perhaps some other extension of $v_i$).

\begin{align*}
\tag{Eisenberg-Gale Relaxation}
\max \quad & \sum_{i\in \cA} \log f_i(\bx_i) \\
& \quad \sum_{i \in \cA} x_{ij} \le 1 && \forall j \in \cG \\
& \quad x_{ij} \ge 0 &&  \forall i \in \cA, j \in \cG \\
\end{align*}

Suppose that $\bx$ is an optimal solution of this relaxation. We will need the following property, which is also stated in \cite{GHV20} in the context of general concave valuations (Lemma 4.1 in \cite{GHV20}). Our proof here is much simpler. 
First, we consider the case of differentiable concave $f_i$ which makes the proof cleaner. (Recall however that $v^+_i$ is not differentiable everywhere.)

\begin{lemma}
\label{lem:local-opt}
For an optimal solution $\bx$ of (\ref{eqns:log-concave}) with differentiable nonnegative monotone concave functions $f_i$, and any other feasible solution $\bx^*$, we have
$$ \sum_{i \in \cA} \frac{f_i(\bx^*_i)}{f_i(\bx_i)} \leq |\cA|. $$
\end{lemma}

\begin{proof}
Since $f_i(\bx)$ is a concave function, we have
$$ f_i(\bx^*_i) \leq f_i(\bx_i) + (\bx^*_i - \bx_i)  \cdot  \nabla f_i(\bx_i).$$
From here, we get
\[
\begin{aligned} \sum_{i \in \cA} \frac{f_i(\bx^*_i)}{f_i(\bx_i)} 
& \leq \sum_{i \in \cA} \frac{f_i(\bx_i) +  (\bx^*_i - \bx_i)  \cdot \nabla f_i(\bx_i)}{f_i(\bx_i)} = |\cA| + \sum_{i \in \cA} (\bx^*_i - \bx_i) \cdot \nabla (\log f_i(\bx_i)) \leq |\cA|
\end{aligned}
\]
using the fact that $\bx^*$ is feasible and $\bx$ is an optimum for the objective function $\sum_{i \in \cA} \log f_i(\bx_i)$.
\end{proof}

To deal with a more general situation where $f_i$ is not necessarily differentiable, and we don't find an exact optimum, we prove a robust version of this lemma.

\begin{lemma}
\label{lem:discrete-opt}
Let $f_i:[0,1]^{\cG} \to \RR$ for each $i \in \cA$ be nonnegative, monotone and concave. For $\eps>0$, let $\bx$ be an $\eps^4$-approximate solution of (\ref{eqns:log-concave}), in the sense that for every other feasible solution $\bx'$,
\begin{equation}
\label{eq:approx-opt}
\sum_{i \in \cA} \log f_i(\bx') \leq \sum_{i \in \cA} (\log f_i(\bx) + \epsilon^4).
\end{equation}
And suppose further that that $y_{ij} \geq \eps$ for all $i,j$, 
Then for every feasible solution $\bx^*$, we have
$$ \sum_{i \in \cA} \frac{f_i(\bx^*_i)}{f_i(\bx_i)} \leq (1 + 2 \eps) |\cA|. $$
\end{lemma}
Note that we must necessarily have $\epsilon \leq 1/|\cA|$, because $1 \geq \sum_{i \in \cA} x_{ij} \geq \eps |\cA|$.

\begin{proof}
Let $\bx$ satisfy the assumptions of the lemma. For any feasible $\bx^*$ and $T \geq 1$, using the concavity of $f_i$, we can write
$$ f_i(\bx^*_i) - f_i(\bx_i) \leq T (f_i(\bx_i + \frac{1}{T} (\bx^*_i - \bx_i)) - f_i(\bx_i)).$$
From here,
$$ \sum_{i \in \cA} \frac{f_i(\bx^*_i) - f_i(\bx_i)}{f_i(\bx_i)} \leq T \sum_{i \in \cA} \frac{f_i(\bx_i + \frac{1}{T} (\bx^*_i - \bx_i)) - f_i(\bx_i)}{f_i(\bx_i)}.$$
Note that since $y_{ij} \geq \eps$, we have $\bx_i + \frac{1}{T} (\bx^*_i - \bx_i) \leq \bx_i + \frac{1}{T} \b1 \leq (1 + \frac{1}{T \eps}) \bx_i$. Also, $f_i(0) \geq 0$, so by monotonicity and concavity, $f_i(\bx_i + \frac{1}{T} (\bx^*_i - \bx_i)) \leq f_i((1+\frac{1}{T\eps}) \bx_i) \leq (1+\frac{1}{T \eps}) f_i(\bx_i)$. Similarly, $f_i(\bx_i + \frac{1}{T} (\bx^*_i - \bx_i)) \geq f_i(\bx_i - \frac{1}{T} \b1) \geq (1-\frac{1}{T \eps}) f_i(\bx_i)$. Hence the ratio $r_i = \frac{f_i(\bx_i + \frac{1}{T} (\bx^*_i - \bx_i)) - f_i(\bx_i)}{f_i(\bx_i)}$ is at most $\delta = \frac{1}{T \eps}$ in absolute value, and we can use the following elementary approximation:
$$ r_i-\delta^2 \leq \log (1+r_i) \leq r_i. $$
Plugging into the inequality above, we obtain
$$ \sum_{i \in \cA} \frac{f_i(\bx^*_i) - f_i(\bx_i)}{f_i(\bx_i)} 
\leq T \sum_{i \in \cA} r_i \leq T \sum_{i \in \cA} ( \delta^2 + \log (1+r_i)) 
= \frac{|\cA|}{T \eps^2} + T \sum_{i \in \cA} \log \frac{f_i(\bx_i + \frac{1}{T} (\bx^*_i - \bx_i))}{f_i(\bx_i)}.$$
Applying the assumption of the lemma to the feasible solution $\bx' = \bx_i + \frac{1}{T} (\bx^*_i - \bx_i)$, 
we have $ \sum_{i \in \cA} \log \frac{f_i(\bx_i + \frac{1}{T} (\bx^*_i - \bx_i))}{f_i(\bx_i)} \leq \epsilon^4 |\cA|,$
which gives
$$ \sum_{i \in \cA} \frac{f_i(\bx^*_i)}{f_i(\bx_i)} = |\cA| + \sum_{i \in \cA} \frac{f_i(\bx^*_i) - f_i(\bx_i)}{f_i(\bx_i)} \leq |\cA| + \frac{|\cA|}{T \eps^2} + T \eps^4 |\cA|.$$
We set $T$ to equate the last two terms: $T = 1 / \eps^3$, which gives the statement of the lemma.
\end{proof}

\begin{corollary}
\label{cor:concave-opt}
\Aviad{I'm not sure what ``fixed'' $\alpha$ means. Maybe we should say that the algorithm runs in time $poly(n,\alpha)$?}
Given a value oracle and a supergradient oracle for each $f_i$, for any constant $\alpha>0$, we can find a solution $\bx$ of (\ref{eqns:log-concave}) in polynomial time such that for any feasible solution $\bx^*$,
$$ \sum_{i \in \cA} \frac{f_i(\bx^*_i)}{f_i(\bx_i)} \leq (1 + \alpha) |\cA|. $$
\end{corollary}

\begin{proof}
For $\eps>0$ (to be chosen at the end), we run a convex optimization algorithm on (\ref{eqns:log-concave}) with the additional constraint that $x_{ij} \geq \eps$, to obtain a solution $\bx$ such that for any feasible $\bx'$ satisfying the same constraint, we have
$$ \sum_{i \in \cA} \log f_i(\bx_i) \geq \sum_{i \in \cA} \log f_i(\bx') - \eps^4 n.$$
By Lemma~\ref{lem:discrete-opt}, this solution also satisfies
$$ \sum_{i \in \cA} \frac{f_i(\bx'_i)}{f_i(\bx_i)} \leq (1+2\eps) n.$$    
Finally, note that every feasible solution $\bx^*$ of (\ref{eqns:log-concave}) can be modified to obtain a feasible solution $\bx' = (1-\eps n) \bx^* + \eps n \cdot \frac{1}{n} \b1$ which satisfies the constraint $x'_{ij} \geq \eps$, and we have $f_i(\bx'_i) \geq (1 - \eps n) f_i(\bx^*_i)$ for every $i \in \cA$. 
Therefore, our solution also satisfies
$$ \sum_{i \in \cA} \frac{f_i(\bx^*_i)}{f_i(\bx_i)} \leq \frac{(1+2\eps) n}{1-\eps n}.$$    
For $\eps = \frac{\alpha}{2+(1+\alpha)n}$, we obtain the desired statement.
\end{proof}

\section{Rematching lemmas}
\label{app:rematching}

Here we prove the rematching lemmas from Section~\ref{sec:rematching}. These are essentially identical to lemmas in previous work on Nash social welfare, only reformulated in a way convenient for our presentation. We give self-contained proofs here for completeness.

\begin{proof}[Proof of Lemma~\ref{lemma:extension}]
Suppose that $S^*_i = H^*_i \cup T^*_i$ where $H^*_i \subseteq \cH$ and $T^*_i \subseteq \cG'$. We define a matching $\pi$ as follows: For each nonempty $H^*_i$, let $\pi(i)$ be the item of maximum value (as a singleton) in $H^*_i$. For $H^*_i=\emptyset$, let $\pi(i)$ be an arbitrary item in $\cH$ not selected as $\pi(i')$ for some other agent. (Since $|\cH| = |\cA|$, we can always find such items.) Recall that $\cA'$ are the agents who get positive value from $\cG'$; in particular, we can assume $T^*_i = \emptyset$ for $i \notin \cA'$.
Then we have, using monotonicity and subadditivity
\begin{eqnarray*}
\prod_{i \in \cA} \frac{v_i(S^*_i)}{\max\{V_i, v_i(\pi(i))\}} 
& \leq & \prod_{i \in \cA} \frac{v_i(T^*_i) + v_i(H^*_i)}{\max\{V_i,v_i(\pi(i))\}}
 \leq \prod_{i \in \cA} \frac{v_i(T^*_i) + |H^*_i| v_i(\pi(i)))}{\max\{V_i,v_i(\pi(i))\}} \\
& \leq & \prod_{i \in \cA'} \left( \frac{v_i(T^*_i)}{V_i} + |H^*_i| \right) \prod_{i \in \cA \setminus \cA'} |H^*_i|.
\end{eqnarray*}
Here we use the AMGM inequality:
\begin{eqnarray*}
\left( \prod_{i \in \cA'} \left( \frac{v_i(T^*_i)}{V_i} + |H^*_i| \right) \prod_{i \in \cA \setminus \cA'} |H^*_i| \right)^{1/|\cA|}
 \leq \frac{1}{|\cA|} \left( \sum_{i \in \cA'} \left( \frac{v_i(T^*_i)}{V_i} + |H^*_i| \right)
 + \sum_{i \in \cA \setminus \cA'} |H^*_i|  \right)
 \leq c + 1
\end{eqnarray*}
where the last inequality is by assumption and the fact that $\sum_{i \in \cA} |H^*_i| \leq |\cH| = |\cA|$. 
\end{proof}

\begin{proof}[Proof of Lemma~\ref{lemma:rematching}]
Let $\tilde{A} = \{i \in \cA: W_i < \max \{ v(\pi(i)), \nu_i \}$.
We define a directed bipartite graph $B$ between $\tilde{A}$ and $\cH$, with two types of edges:
$E_\tau = \{ (\tau(i), i): i \in \tilde{A}\}$ and $E_\pi = \{ (i,\pi(i): i \in \tilde{A} \}$.
We also define:
\begin{itemize}
    \item $A_\nu = \{ i \in \tilde{A}: \nu_i > v_i(\pi(i)) \} $,
    \item $A_\tau = A_\nu \cup \{ i \in \tilde{A}: \exists \mbox{ directed path in } B \mbox{ from } i \mbox{ to } A_\nu \}$,
    \item $A_\pi = \tilde{A} \setminus A_\tau$.     
\end{itemize}

We define a matching $\rho$ as follows;
\begin{itemize}
    \item For $i \in A_\tau$, $\rho(i) := \tau(i)$,
    \item For $i \in A_\pi$, $\rho(i) := \pi(i)$.
    \item For $i \notin \tilde{A}$, we define $\rho(i)$ arbitrarily, to make $\rho$ a matching.
\end{itemize}

First, observe that this is indeed a matching:
If it was the case that $\tau(i) = \pi(i') = j$ for some $i \in A_\tau, i' \in A_\pi$,
then we would have edges $(i',j)$ and $(j, i)$ in the graph, and since there is a directed path from $i$ to $A_\nu$ ($i \in A_\tau$), there would also be a directed path from $i'$ to $A_\nu$, contradicting the fact that $i' \in A_\pi$. Hence, $\rho$ is a matching.

Also, it's easy to see that for every

Next, we analyze the value guarantee for $\rho$:
$$ \prod_{i \in \cA} \max\{W_i, v_i(\rho(i)) \} \geq 
\prod_{i \in \cA \setminus \tilde{A}} W_i \prod_{i \in \tilde{A}} v_i(\rho(i)) = 
\prod_{i \in \cA \setminus \tilde{A}} W_i \prod_{i \in A_\tau} v_i(\tau(i)) \prod_{i \in A_\pi} v_i(\pi(i)). $$
We claim that $\prod_{i \in A_\tau} v_i(\tau(i)) \geq \prod_{i \in A_\nu} \nu_i \prod_{i \in A_\tau \setminus A_\nu} v_i(\pi(i))$. Observe that the vertices of $A_\tau$ can be covered disjointly by directed paths that terminate in $A_\nu$ (from each vertex of $A_\tau$, there is such a path and it is also unique, because the in-degrees / out-degrees in the graph are at most $1$). Let $P$ denote the $A_\tau$-vertices on some directed path like this, and let $s$ be its last vertex (in $A_\nu$). If it was the case that $\prod_{i \in P} v_i(\tau(i)) < \nu_s \prod_{i \in P \setminus \{s\}} v_i(\pi(i)) $, then we could modify the matching $\tau$ by swapping its edges on $P$ for the $\pi$-edges from $P \setminus \{s\}$, and finally an element of value $\nu_i$ for $s$ (since this item is outside of $\cH$ and hence available). This would increase the value of the matching $\tau$, which was chosen to be optimal, so this cannot happen.

It follows that $\prod_{i \in P} v_i(\tau(i)) \geq \nu_s \prod_{i \in P \setminus \{s\}} v_i(\pi(i))$     for every maximal directed path terminating in $A_\nu$, and since these paths cover $A_\tau$ disjointly, by combining all these inequalities we obtain 
$$ \prod_{i \in A_\tau} v_i(\tau(i)) \geq \prod_{i \in A_\nu} \nu_i \prod_{i \in A_\tau \setminus A_\nu} v_i(\pi(i)).$$
Substituting this into the inequality above,
$$ \prod_{i \in \cA} \max\{W_i, v_i(\rho(i)) \} \geq 
\prod_{i \in \cA \setminus \tilde{A}} V_i \prod_{i \in A_\nu} \nu_i \prod_{i \in A_\tau} v_i(\pi(i)) = \prod_{i \in \cA} \max \{ W_i, \nu_i, v_i(\pi(i)) \}.$$
\end{proof}
\section{Concentration of subadditive functions}
\label{app:concentration}

Let us start with a simple lower bound on the expected value of a random set with independently sampled elements.

\begin{lemma}
\label{lemma:subadditive-exp}
If $f:2^M \to \RR_+$ is a monotone subadditive function and $R$ is a random subset of $S$ where each element appears independently with probability $1/k$, $k\geq 1$ integer, then
$$ \E[f(R)] \geq \frac{1}{k} f(S).$$ 
\end{lemma}

\begin{proof}
Consider a random coloring of $S$, where every element $j \in S$ receives independently a random color $c(j) \in [k]$. Defining $S_\ell = \{ j \in S: c(j) = \ell \}$,
we see that each set $S_\ell$ has the same distribution as the set $R$ in the Lemma.
Therefore,
$$ \E[f(R)] = \E[f(S_1)] = \ldots = \E[f(S_k)] =  \E\left[ \frac{1}{k} \sum_{\ell=1}^{k} f(S_\ell) \right] \geq \frac{1}{k} f(S) $$
by subadditivity.
\end{proof}

This property is similar to the properties of {\em submodular} or {\em self-bounding functions}, which satisfy very convenient concentration bounds (similar to additive functions). Unfortunately, the same bounds are not true for subadditive functions; however, some concentration results can be recovered with a loss of certain constant factors.

Here we state a powerful concentration result presented by Schechtman \cite{Schechtman99}, based on the ``$q$-point control'' concentration inequality by Talagrand \cite{Talagrand89,Talagrand95}. We state it here in a simplified form suitable for our purposes.

\begin{theorem}
\label{thm:Schechtman}
Let $f:2^M \to \RR_+$ be a monotone subadditive function, where $f(\{i\}) \leq 1$ for every $i \in M$. Then for any real $a > 0$ and integers $k,q \geq 1$, and a random set $R$ from a product distribution,
$$ \Pr[f(R) \geq (q+1) a + k] \ (\Pr[f(R) \leq a])^q \leq \frac{1}{q^k}.$$
\end{theorem}

This statement can be obtained from Corollary 12 in \cite{Schechtman99} by extending the definition of $f$ to $\Omega^* = \bigcup_{I \subset M} 2^I$ simply by saying $f_I(S) = f(S)$ for all $S \subseteq I$. Also, we identify $2^I$ with $\{0,1\}^I$ in a natural way. Assuming $f(\{i\}) \leq 1$ means that $0 \leq f(S+i) - f(S) \leq 1$ for any set $S$, by monotonicity and subadditivity. Therefore, $f$ is $1$-Lipschitz with respect to the Hamming distance, as required in \cite{Schechtman99}. The statement holds for any product distribution, i.e.~a random set $R$ where elements appear independently.

Note that Theorem~\ref{thm:Schechtman} refers to tails on both sides and hence is more convenient to use with the median of $f$ than the expectation. The next lemma shows that this is not a big issue, since the theorem also implies that the median and expectation must be within a constant factor.

\begin{definition}
We define the median of a random variable $Z$ as any number $\med(Z) = m$ such that
$$ \Pr[Z \leq m] \geq 1/2, \ \ \ \Pr[Z \geq m] \geq 1/2.$$
\end{definition}

For any nonnegative variable, obviously $\E[Z] \geq \frac12 \med(Z)$.
For subadditive functions of independent random variables, we also get a bound in the opposite direction.

\begin{lemma}
\label{lemma:median}
Let $f:2^M \to \RR_+$ be a monotone subadditive function, where $f(\{i\}) \leq 1$ for every $i \in M$. Then for a random set $R$ from a product distribution,
$$ \E[f(R)] \leq 5(\med(f(R))+1).$$
\end{lemma}

\begin{proof}
Let $a = \med(f(R))$ be the median. We apply Theorem~\ref{thm:Schechtman} with $k = q+1, q \geq 3$:
$$ \Pr[f(R) \geq (q+1)(a+1)] \leq \frac{2^q}{q^{q+1}} \leq \left( \frac{2}{3} \right)^q. $$
We can bound the expectation as follows:
$$ \E[f(R)] \leq 4(a+1) + (a+1) \sum_{q=3}^{\infty} \Pr[f(R) > (q+1)(a+1)] 
 \leq 4(a+1) + (a+1) \sum_{q=3}^{\infty} \left( \frac{2}{3} \right)^q < 5(a+1). $$
\end{proof}

Hence, we obtain the following as a corollary of Theorem~\ref{thm:Schechtman} and Lemma~\ref{lemma:median}. For convenience, we also introduce a parameter $\nu>0$ as an upper bound on singleton values.

\begin{theorem}
\label{thm:lower-tail}
Let $f:2^M \to \RR_+$ be a monotone subadditive function, where $f(\{i\}) \leq \nu$, $\nu >0$, for every $i \in M$. Then for any integers $k,q \geq 1$, and a random set $R$ where elements appear independently,
$$ \Pr\left[f(R) \leq \frac{\E[f(R)]}{5(q+1)} - \frac{(k+1) \nu}{q+1} \right] \leq \left(\frac{2}{q^k}\right)^{1/q}.$$
\end{theorem}

\begin{proof}
Assume first $g$ is a function satisfying the assumptions with $\nu=1$.
We use Theorem~\ref{thm:Schechtman} with parameter $a = (\med(g) - k) / (q+1)$.
Note that $\Pr[g(R) \geq (q+1)a + k] = \Pr[g(R) \geq \med(g)] = 1/2$. Hence, Theorem~\ref{thm:Schechtman} gives
$$ \frac{1}{2} \cdot (\Pr[g(R) \leq a])^q \leq \frac{1}{q^k}.$$
From Lemma~\ref{lemma:median}, we have $a = (\med(g) - k) / (q+1) \geq (\frac{1}{5} \E[g(R)] - 1-k) / (q+1)$ which implies the statement for $\nu=1$:
$$ \Pr\left[g(R) \leq \frac{\E[g(R)]}{5(q+1)} - \frac{k+1}{q+1} \right] \leq \left(\frac{2}{q^{k}}\right)^{1/q}.$$
For a function $f$ satisfying the assumptions for $\nu > 0$, we apply this inequality to the function $g(R) = \frac{1}{\nu} f(R)$, to obtain the statement of the lemma.
\end{proof}


\label{LastPage}
\end{document}